\documentclass[
aps,
preprint,
nofootinbib,
amsmath,amssymb,
onecolumn
]{revtex4-2}


\usepackage[utf8]{inputenc}
\usepackage{graphicx}
\usepackage{color}
\usepackage[
	bookmarks=false,
	colorlinks=true,
	citecolor=refcolor,
	linkcolor=refcolor,
	urlcolor=magenta
]{hyperref}
\usepackage{bm}
\usepackage{braket}
\newcommand{\Tr}{\mathrm{Tr}}

\usepackage{xcolor}
\definecolor{refcolor}{rgb}{0.3,0.3,1}
\usepackage{colortbl}

\usepackage{amsthm}
\newtheorem{theorem}{Theorem}[section]
\newtheorem{proposition}{Proposition}[section]
\newtheorem{lemma}{Lemma}[section]
\newtheorem{corollary}{Corollary}[section]
\newtheorem{definition}{Definition}[section]
\newtheorem*{question}{Question}

\usepackage{tcolorbox}
\tcolorboxenvironment{theorem}{colback= lightgray!10,}
\tcolorboxenvironment{proposition}{colback= lightgray!10,}
\tcolorboxenvironment{lemma}{colback= lightgray!10,}
\tcolorboxenvironment{corollary}{colback= lightgray!10,}
\tcolorboxenvironment{definition}{colback= lightgray!10,}
\tcolorboxenvironment{question}{colback= lightgray!10,}

\usepackage{mathtools}
\usepackage{array}
\usepackage{comment}

\usepackage{titlesec}

\renewcommand{\thesection}{\arabic{section}}
\renewcommand{\thesubsection}{\thesection.\arabic{subsection}}
\makeatletter
\renewcommand{\p@subsection}{}
\makeatother
\titleformat*{\section}{\large\bfseries}
\titleformat*{\subsection}{\bfseries}

\usepackage[format=hang, justification=raggedright]{caption}

\begin{document}

\pagenumbering{gobble}

\title{
	\fontsize{15pt}{15pt}\selectfont
    Application of resource theory
    based on free Clifford+$kT$ computation
    to early fault-tolerant quantum computing
    \vspace{6pt}
}

\author{Yuya O. Nakagawa}
\email{nakagawa@qunasys.com}
\author{Yasunori Lee}

\affiliation{\vspace{12pt}\small QunaSys,
1-13-7 Hakusan, Bunkyo, Tokyo 113-0001, Japan}

\begin{abstract}
Recent advances in quantum hardware are bringing fault-tolerant quantum computing (FTQC) closer to reality.
In the early stage of FTQC,
however, the numbers of available logical qubits and high-fidelity $T$ gates remain limited, making it crucial to optimize the quantum resource usage.
In this work, we aim to study the simulation cost of general quantum states under the constraint that only $k$ $T$ gates can be used, alongside an unlimited number of Clifford gates.
Inspired by the notion of robustness of magic (RoM) which quantifies the cost of quantum-circuit simulation using stabilizer states ($k = 0$), we introduce its generalization, which we call Clifford+$kT$ robustness, treating Clifford+$kT$ states as free resources.
We explore theoretical properties of Clifford+$kT$ robustness and in particular derive a lower bound that reveals the (in)efficiency of quantum-circuit simulation using Clifford+$kT$ states.
Through numerical computations, we also evaluate Clifford+$kT$ robustness for key resource states for universal quantum computation, such as tensor products of the magic states.
Our results allow to assess the sampling-cost reduction achieved by the use of Clifford+$kT$ states instead of stabilizer states, providing practical guidance for efficient resource usage in the early-FTQC era.
\end{abstract}

\maketitle

\newpage

\tableofcontents

\newpage

\pagenumbering{arabic}

\section{Introduction}\label{sec:intro}
Research and development of quantum computing hardware have progressed rapidly over the past decade, leading to the demonstration of various quantum computer prototypes.
While almost all quantum devices realized so far have been noisy and lacked error correction, recent years have witnessed a growing number of reports possibly leading to the realization of scalable fault-tolerant quantum computation (FTQC).\footnote{
	See \textit{e.g.} \cite{Régent:2025} for a comprehensive list of notable experiments.
}
As a result, increasing attention is being directed toward the early-FTQC era.
In this regime, gate fidelities and the number of available logical qubits remain limited, making it crucial to devise methods for performing meaningful computations with the scarce resources and to estimate the resource requirements for concrete problems in practical applications.\footnote{
	See \textit{e.g.} \cite{kuroiwa2023averaging} and references therein for a partial list of related works.
}

In general FTQC architectures, a quantum circuit is decomposed into elementary gates which can be implemented in a fault-tolerant manner.
One of the most widely used sets of such gates consists of Clifford gates and a $T$ gate.
The Gottesman-Knill theorem~\cite{gottesman1998heisenberg} states that quantum circuits composed solely of the former (and measurements in the computational basis) is efficiently simulatable on classical computers, and thus the power of quantum computation beyond its classical counterpart is necessarily attributed to the latter.
At the same time, the implementation of $T$~gates is in general significantly more costly than that of Clifford gates also on quantum devices,
requiring \textit{e.g.} magic state distillation~\cite{Knill:2004,BravyiKitaev:2004} which consumes a large number of clock cycles and physical qubits.\footnote{
	\setlength{\baselineskip}{16pt}%
	However, recent proposals such as zero-level magic state distillation~\cite{itogawa2025} and magic state cultivation~\cite{gidney2024magic} might reduce the cost of $T$-gate implementation.
}
Therefore, the number of necessary $T$~gates typically serves as a key metric for the cost of a (fault-tolerant) quantum algorithm.

Due to the hardware constraints, the number of executable $T$ gates is expected to be severely limited in the early-FTQC era.
This motivates us to explore ways to carry out the desired quantum computations without increasing the number of $T$ gates,
possibly with the help of classical high-performance computing resources,
which is of critical importance to utility of quantum computers at their early stage.
More concretely, we are led to the following question:

\newpage

\begin{question}
	Given the ability to perform quantum computations involving only $k$ $T$~gates (with an unlimited number of Clifford gates),
	how can we simulate quantum circuits/states that require more than $k$ $T$ gates?
\end{question}

The $k=0$ case corresponds to the ordinary classical simulation 
using only Clifford gates, which has already been studied for a long time, notably since the work of \cite{AaronsonGottesman:2004}.
Along the way, two major approaches have been established, respectively based on
\begin{itemize}
\item stabilizer rank~\cite{BravyiSmithSmolin:2015, BravyiGosset:2016} and 
\item (negativity of) quasi-probability representation over stabilizer states~\cite{PashayanWallmanBartlett:2015, HowardCampbell:2016} 
\end{itemize}
of the target quantum state.
While the former approach actually arose from a qubit version of the above Question \cite{BravyiSmithSmolin:2015} and shares essentially the same spirit as ours,
it is known that stabilizer rank of a given quantum state is usually not easy to compute.
On the other hand, the latter approach is blessed with (at least relative) computability,
where the value of a quantity called the \emph{robustness of magic (RoM)} characterizes the cost of simulation via sampling over stabilizer states~\cite{HowardCampbell:2016}. 
In this work, we generalize the notion of RoM to extend its framework to $k \geq 1$ cases by replacing stabilizer states with what we call \emph{Clifford+$kT$ states} (a precise definition will be provided in Section~\ref{sec:notations}),
and investigate, for example, the cost of simulating a general quantum state via sampling over the Clifford+$kT$ states.

Here we note that the notion of Clifford+$kT$ gates/states has also appeared in previous works on \textit{e.g.}
\begin{itemize}
	\item \textit{learning} of a given unknown Clifford+$kT$ gate/state via process/state tomography\\
	\cite{lai2022, leone2024learningefficient, grewal2023efficient, Leone2024learningtdoped, hangleiter2024, ChiaLaiLin:2023}, or
	\item approximate \textit{design} of (pseudo-)random unitaries/states by Clifford+$kT$ gates/states\\
	\cite{HaferkampMontealegreMoraHeinrichEisertGrossRoth:2020, GrewalIyerKretschmerLiang:2024_pseudorandom_improvedbound, MagniChristopoulosDeLucaTurkeshi:2025, LeoneOlivieroHammaEisertBittel:2025, ZhangVijayGuBao:2025},
\end{itemize}
often under the name of \textit{$t$-doped Clifford unitaries/stabilizer states}
(where $t$ corresponds to our $k$, and in some cases allows generic non-Clifford gates in place of $T$ gates).
Although the foci of these works are somewhat different from ours, analysis in this work might also shed new light on these studies.

\newpage

The contributions and structure of this paper are as follows:
\begin{itemize}
	\item In Section~\ref{sec:notations}, we introduce our notation and formally define \emph{Clifford+$kT$ gates/states}.
	
	\item In Section~\ref{sec:CkT robustness}, we define \emph{Clifford+$kT$ robustness}, a resource-theoretic quantity treating Clifford+$kT$ states as free, and investigate its properties
	basically in the same way as the case of \emph{robustness of magic (RoM)}~\cite{HowardCampbell:2016}. In particular, 
	\begin{itemize}
		\item[$\circ$] We show that Clifford+$kT$ robustness of a general state characterizes the simulation sampling cost of the state.
		\item[$\circ$] We also derive a lower bound of Clifford+$kT$ robustness 
			and demonstrate an operational meaning of Clifford+$kT$ robustness by considering state conversion under Clifford+$kT$ channels. 
	\end{itemize}
	
	\item In Section~\ref{sec:numerics}, 
	we enumerate $n$-qubit Clifford+$kT$ states for small $n,k$ and use them to calculate Clifford+$kT$ robustness of various non-Clifford states including the tensor products of the magic states.
	In particular,
	\begin{itemize}
		\item[$\circ$] We find examples where the additional use of Clifford+$kT$ states succeeds in lowering the sampling cost compared to the ordinary stabilizer-based simulation.
		\item[$\circ$] We also find examples where an increase in $k$ does not immediately lead to reduction in Clifford+$kT$ robustness (and hence in the simulation sampling cost).
		\item[$\circ$] As a by-product, we also provide normal forms for single-qubit Clifford+$kT$ \emph{states} (which resembles the known normal forms for single-qubit Clifford+$kT$ \emph{gates} \cite{MatsumotoAmano:2008,GilesSelinger:2013,GossetKliuchnikovMoscaRusso:2013}).
	\end{itemize}
	These results would give some hints toward simulation strategies for optimal use of the limited number of $T$ gates in the early-FTQC era.
	
	\item In Section~\ref{sec:conclusion}, we present the outlook and possible future directions.
\end{itemize}

\newpage

\section{Preliminaries} \label{sec:notations}
Let us begin by summarizing notation and terminology to be used.

\subsection{$T$ gate and Clifford gates}\label{subsec:t_and_clifford_gates}
A \emph{Pauli gate/operator} $P$ acting on $n$ qubits is defined as a tensor product of $n$ one-qubit Pauli operators $P^{(q)}$ each acting on a $q$-th qubit
\begin{equation*}
 P =\bigotimes_{q=1}^n P^{(q)} \quad \Big(P^{(q)} \in \{I,X,Y,Z\}\Big),
\end{equation*}
where the four gates are respectively
\begin{equation*}
		I = \left(
			\begin{array}{wc{16pt}wc{16pt}}
				1 & 0\\
				0 & 1
			\end{array}
		\right),
		\quad
		X = \left(
			\begin{array}{wc{16pt}wc{16pt}}
				0 & 1\\
				1 & 0
			\end{array}
		\right),
		\quad
		Y = \left(
			\begin{array}{wc{16pt}wc{16pt}}
				0 & -i\\
				i & 0
			\end{array}
		\right),
		\quad
		Z = \left(
			\begin{array}{wc{16pt}wc{16pt}}
				1 & 0\\
				0 & -1
			\end{array}
		\right)
\end{equation*}
in terms of the computational basis
\begin{equation*}
	\ket{0} = \left(
		\begin{array}{c}
			1\\
			0
		\end{array}
	\right),
	\quad
	\ket{1} = \left(
		\begin{array}{c}
			0\\
			1
		\end{array}
	\right).
\end{equation*}
Note that we do not include global phase factors $\pm 1, \pm i$ in the definition of Pauli operators.\footnote{
	\setlength{\baselineskip}{16pt}%
	Including these global phase factors, Pauli operators form a ($n$-qubit Pauli) group under ordinary multiplication.
}
Since $P^2 = I^{\otimes n}$ for any Pauli operator $P$, a rotation $R_P(\theta)$ generated by $P$ is given by
\begin{equation}\label{eq:pauli-rotation-def}
	R_P(\theta)
	\coloneqq \exp(-i \tfrac{\theta}{2} P)
	= \cos \tfrac{\theta}{2} \cdot I^{\otimes n} - i\sin\tfrac{\theta}{2} \cdot P.
\end{equation}
The \emph{$T$ gate}, which is of our primary interest in this paper, is a one-qubit unitary operator
\begin{equation*}
		T = \left(
			\begin{array}{wc{16pt}c}
				1 & 0\\
				0 & e^{2\pi i/ 8}
			\end{array}
		\right)
\end{equation*}
and is in fact equivalent to a Pauli rotation $R_Z(\frac{\pi}{4})$ up to a global phase factor, namely, $T = e^{+i\pi/8} R_Z(\frac{\pi}{4})$
(or letting $\sim$ denote the equivalence up to a global phase factor, $T \sim R_Z(\tfrac{\pi}{4})$).

\newpage

An $n$-qubit operator $C$ such that either $CPC^\dagger$ or $-CPC^\dagger$ is a Pauli operator whenever $P$ is a Pauli operator is called a \emph{Clifford gate/operator}.\footnote{
	\setlength{\baselineskip}{16pt}%
	In other words, $n$-qubit Clifford gates/operators are the elements of the normalizer of the $n$-qubit Pauli group in $U(2^n)$.
	It is known that those with entries in $\mathbb{Q}[e^{2\pi i/8}]$ form a group
	whose quotient by a group generated by the $n$-qubit Pauli group and $e^{2\pi i/8}I$ is isomorphic to $Sp(2n,\mathbb{Z}/2\mathbb{Z})$ \cite{CalderbankRainsShorSloane:1996}.
	Therefore, the number of $n$-qubit Clifford gates up to global phase factors is given by $4^{n} \cdot |Sp(2n,\mathbb{Z}/2\mathbb{Z})|$
	where $4^{n}$ is the number of $n$-qubit Pauli operators and $|Sp(2n,\mathbb{Z}/2\mathbb{Z})| = 2^{n^2} \prod_{j=1}^n(4^j - 1)$. 
	This indeed reproduces the number for the $n=1$ case as $4^1 \cdot 6 = 24$.
}
From Eq.~\eqref{eq:pauli-rotation-def}, it transforms a Pauli rotation operator as
\begin{equation} \label{eq:CR_PC=R_P'}
	CR_P(\theta) C^\dagger = R_{CPC^\dagger}(\pm\theta).
\end{equation}
Clifford gates constitute a group under multiplication
which is generated by the following Hadamard, $S$, and CNOT gates, up to global phase factors 
(see \textit{e.g.} \cite{Gottesman:1997}):
\begin{equation*}
	H = \frac{1}{\sqrt{2}} \left(
		\begin{array}{wc{16pt}wc{16pt}}
			1 & 1\\
			1 & -1
		\end{array}
	\right),
	\quad
	 S = \left(
 		\begin{array}{wc{16pt}wc{16pt}}
 			1 & 0\\
 			0 & i
 		\end{array}
 	\right),
	 \quad
	 \mathrm{CNOT} =
	 \ket{0}\!\bra{0} \otimes I  + \ket{1}\!\bra{1} \otimes X,
\end{equation*}
where the latter is a two-qubit gate and can act on any pair of qubits.

For example for the $n=1$ case, the two gates $H$ and $S$ indeed transform each Pauli gate as
\begin{align}\label{eq:one-qubit_clifford_transformations}
		HXH^{\dagger} & = Z, & H\,YH^{\dagger} & = -Y, & HZH^{\dagger} & = X, \nonumber\\
		SXS^{\dagger} & = Y, & S\,YS^{\dagger} & = -X, & SZS^{\dagger} & = Z.
\end{align}
Also, from the relations
\begin{align*}
	H^2 &= I,\\
	S^4 &= I,\\
	(SH)^3 &= e^{2\pi i/8}I,
\end{align*}
one can verify that, up to global phase factors, they generate the symmetric group $S_4$
and hence there are $|S_4| = 24$ one-qubit Clifford gates in total. 
In particular, the Pauli gates are (Clifford gates by definition and) generated by $H$ and $S$ gates as
\begin{align*}
	X &= HSSH,\\
	Y &= SHSSHSSS,\\
	Z &= SS,
\end{align*}
as can be seen from Eq.~\eqref{eq:one-qubit_clifford_transformations}.

\newpage

\subsection{Stabilizer states and magic states}
It is known that an $n$-qubit state $\ket{\psi}$ can be constructed by applying suitable Clifford gate(s) to $\ket{0}^{\otimes n}$
if and only if $\ket{\psi}$ is \textit{stabilized} by (\textit{i.e.} an eigenstate $P_i\ket{\psi} = \pm\ket{\psi}$ of) $2^n$~Pauli gates~\cite{AaronsonGottesman:2004}.
Such states are called the \emph{stabilizer states}.
For example for the $n=1$ case, there are six stabilizer states in total (up to a global phase factor);
they are stabilized as
\begin{align*}
		X \tfrac{\ket{0}+\ket{1}}{\sqrt{2}} & = +\tfrac{\ket{0}+\ket{1}}{\sqrt{2}}, &
		X \tfrac{\ket{0}-\ket{1}}{\sqrt{2}} & = -\tfrac{\ket{0}-\ket{1}}{\sqrt{2}},\\
		Y \tfrac{\ket{0}+i\ket{1}}{\sqrt{2}} & = +\tfrac{\ket{0}+i\ket{1}}{\sqrt{2}}, &
		Y \tfrac{\ket{0}-i\ket{1}}{\sqrt{2}} & = -\tfrac{\ket{0}-i\ket{1}}{\sqrt{2}},\\
		Z \ket{0} & = +\ket{0}, &
		Z \ket{1} & = -\ket{1},
\end{align*}
in addition to $I\ket{\psi} = +\ket{\psi}$ for all states $\ket{\psi}$ and constructed as
\begin{align*}
		H\ket{0} & = \tfrac{\ket{0}+\ket{1}}{\sqrt{2}} \eqqcolon \ket{+}, &
		HX\ket{0} & = \tfrac{\ket{0}-\ket{1}}{\sqrt{2}} \eqqcolon \ket{-},\\
		SH\ket{0} & = \tfrac{\ket{0}+i\ket{1}}{\sqrt{2}}, &
		SHX\ket{0} & = \tfrac{\ket{0}-i\ket{1}}{\sqrt{2}},\\
		\ket{0} &,&
		X\ket{0} & = \ket{1}.
\end{align*}

Among one-qubit pure states,
there exist two non-stabilizer states $\ket{H}$ and $\ket{SH}$
of special importance called the \textit{magic states} \cite{BravyiKitaev:2004} defined by
\begin{align}
	\ket{H}\!\bra{H} & = \frac{1}{2}\left(I + \frac{X+Z}{\sqrt{2}}\right), \label{eq:def of H state} \\[6pt]
	\ket{SH}\!\bra{SH} & = \frac{1}{2}\left(I + \frac{X+Y+Z}{\sqrt{3}}\right). \label{eq:def of SH state}
\end{align}
As the notation shows, they are eigenstates of $H$ and $SH$ respectively,
which follows \textit{e.g.} from Eq.~\eqref{eq:one-qubit_clifford_transformations}.
In particular, the former can be written more explicitly as\footnote{
	\setlength{\baselineskip}{16pt}%
	Using $T^\dagger XT = \frac{X-Y}{\sqrt{2}}$, one can immediately verify
	\begin{equation*}
		HTHS \left[\frac{1}{2}\left(I + \frac{X+Z}{\sqrt{2}}\right)\right] S^\dagger HT^\dagger H
		=
		HT \left[\frac{1}{2}\left(I + \frac{-Y+X}{\sqrt{2}}\right)\right] T^\dagger H
		=
		H \left[\frac{1}{2}\left(I + X\right)\right] H
		=
		\ket{0}\!\bra{0}.
	\end{equation*}
}
\begin{equation*}
	\ket{H} \sim S^\dagger H T^\dagger H\ket{0}.
\end{equation*}
This shows that $\ket{H}$ is in fact equivalent to $T\ket{+}$ (up to Clifford gates)
which serves as a resource state for applying $T$ gates to a given quantum circuit by gate teleportation~\cite{ZhouLeungChuang:2000}.

\newpage

\subsection{Clifford+$kT$ gates/states}\label{subsec:Clifford+kt_gates/states}

We first define Clifford+$kT$ gates.
\begin{definition}[Clifford+$kT$ gates]\label{def:Clifford+kT_gates}
	An $n$-qubit gate $U$ is called a Clifford+$kT$ gate if it can be written as
	\begin{equation*}
		U = C_k T^{(q_k)} C_{k-1} T^{(q_{k-1})} \ldots C_1 T^{(q_1)} C_0
	\end{equation*}
	where $C_i$'s are ($n$-qubit) Clifford gates and $T^{(q_i)}$'s are $T$ gates acting on a $q_i$-th qubit. 
\end{definition}
\noindent
Since $T \sim R_{Z}(\frac{\pi}{4})$, Eq.~\eqref{eq:CR_PC=R_P'} allows us to rewrite $U$ as
\begin{align}\label{eq:Clifford+kT_gate_seq_rot_normal_form}
	&
	C_k T^{(q_k)}
	\underbrace{C_k^\dagger C_k}_{I}
	C_{k-1} T^{(q_{k-1})}
	\underbrace{(C_{k} C_{k-1})^\dagger C_{k} C_{k-1}}_{I} \cdots \notag\\
	\sim \  &
	R_{P_k}(\pm\tfrac{\pi}{4})
	R_{P_{k-1}}(\pm\tfrac{\pi}{4})
	\cdots
	R_{P_1}(\pm\tfrac{\pi}{4})
	\underbrace{C_k C_{k-1} \cdots C_1 C_0}_{C},
\end{align}
namely a Clifford gate $C$ followed by $k$ sequential non-trivial Pauli rotations of $\theta = \pm\frac{\pi}{4}$,
up to a global phase factor.
We will later utilize this expression (strictly speaking its improved version \cite{GossetKliuchnikovMoscaRusso:2013}) in numerical enumeration of Clifford+$kT$ states (Sec.~\ref{sec:numerics}).

Let us call $U$ a \emph{strict} Clifford+$kT$ gate if $U$ is not a Clifford+$k'T$ gate for any $k' < k$.
For $n=1$, it is known that a strict Clifford+$kT$ gate can be uniquely expressed (up to a global phase factor)
in the Matsumoto-Amano normal form~\cite{MatsumotoAmano:2008, GilesSelinger:2013} as
\begin{equation*}
	(\varepsilon|T)(HT|SHT)^\ast C
\end{equation*}
with $k$ $T$ gates appearing in total (and $C$ is a Clifford gate),
or in other words (for $k \geq 1$)
\begin{equation}\label{eq:Matsumoto-Amano_normal_form}
	\begin{array}{rl}
		T\,(HT|SHT)\{k-1\}\,C,\\
		HT\,(HT|SHT)\{k-1\}\,C, &\text{ or}\\
		SHT\,(HT|SHT)\{k-1\}\,C,
	\end{array}
\end{equation}
both in regular expression.
It is guaranteed that any two distinct expressions correspond to different strict Clifford+$kT$ gates,
and therefore there are $3 \cdot 2^{k-1} \cdot 24$ of them in total.

Here, we also mention some basic properties of strict Clifford+$kT$ gates for later use:

\newpage

\begin{proposition}\label{prop:hermite_conjugate}
	For any $n$-qubit strict Clifford+$kT$ gate $U$,
	its Hermite conjugate $U^\dagger$ is also a strict Clifford+$kT$ gate.
\end{proposition}

\begin{proof}
	Suppose $U^\dagger$ is a strict Clifford+$k'T$ gate.
	From Definition \ref{def:Clifford+kT_gates}, it follows that $k' \leq k$.
	The same is true for $(U^\dagger)^\dagger = U$, which implies $k \leq k'$ (and hence $k'=k$).
\end{proof}

\begin{proposition}\label{prop:adding_Clifford}
	For any $n$-qubit strict Clifford+$kT$ gate $U$ and any Clifford gate $C$,
	gates $UC$ and $CU$ are also strict Clifford+$kT$ gates.
\end{proposition}

\begin{proof}
	Suppose $UC$ (resp. $CU$) is a strict Clifford+$k'T$ gate.
	From Definition \ref{def:Clifford+kT_gates}, it follows that $k' \leq k$.
	The same is true for $U = (UC)C^\dagger$ (resp. $C^\dagger(CU)$), which implies $k \leq k'$.
\end{proof}

\begin{proposition}\label{prop:adding_T}
	For any $n$-qubit strict Clifford+$kT$ gate $U$,
	a gate $T^{(q_i)}U$ is
	\begin{itemize}
	\item a strict Clifford$+(k+1)T$ gate,
	\item a strict Clifford+$kT$ gate, or
	\item a strict Clifford$+(k-1)T$ gate.
	\end{itemize}
	The same is true for a gate $UT^{(q_i)}$.
\end{proposition}

\begin{proof}
	Suppose $U' = T^{(q_i)}U$ (resp. $UT^{(q_i)}$) is a strict Clifford+$k'T$ gate.
	From Definition~\ref{def:Clifford+kT_gates}, it follows that $k' \leq k+1$.
	The same argument applied to $S^{(q_i)}U = T^{(q_i)}U'$ (resp. $US^{(q_i)}$) which is a strict Clifford+$kT$ gate (from Proposition~\ref{prop:adding_Clifford}) implies $k \leq k'+1$.
\end{proof}

\noindent
In fact, the second possibility in Proposition~\ref{prop:adding_T} is ruled out (at least) for $n=1$:
\begin{proposition}\label{prop:adding_T_one_qubit}
	For any one-qubit strict Clifford+$kT$ gate $U$,
	a gate $TU$ is either
	\begin{itemize}
	\item a strict Clifford$+(k+1)T$ gate or
	\item a strict Clifford$+(k-1)T$ gate.
	\end{itemize}
	The same is true for a gate $UT$.
\end{proposition}

\begin{proof}
	The $U'= TU$ case is directly verifiable from the Matsumoto-Amano normal form.
	Combining with Proposition~\ref{prop:hermite_conjugate},
	the same is true for both $TU^\dagger$ and $(TU^\dagger)^{\dagger} = UTS^\dagger$,
	and Proposition~\ref{prop:adding_Clifford} further confirms that the same is true for $UT$.
\end{proof}

\newpage

We then define our main topic of study, Clifford+$kT$ states.
\begin{definition}[Clifford+$kT$ states]\label{def:Clifford+kT_states}
	An $n$-qubit state $\ket{\psi}$ is called a Clifford+$kT$ state if there exists a Clifford+$kT$ gate $U$ such that
	\begin{equation*}
		\ket{\psi} = U \ket{0}^{\otimes n}.
	\end{equation*}
\end{definition}
\noindent
Let us call $\ket{\psi}$ a \emph{strict} Clifford+$kT$ state if $\ket{\psi}$ is not a Clifford+$k'T$ state for any $k' < k$.
We let $\mathcal{C}_{n}^{+(\leq k)T}$ (resp. $\mathcal{C}_{n}^{+kT}$) denote the set of $n$-qubit non-strict (resp. strict) Clifford+$kT$ states: 
\begin{equation*}
	\mathcal{C}_{n}^{+kT}
	=
	\mathcal{C}_{n}^{+(\leq k)T}\ \big\backslash \bigcup_{0\leq k'\leq k-1}\mathcal{C}_{n}^{+(\leq k')T}.
\end{equation*}
For example, the set of stabilizer states is equal to $\mathcal{C}_{n}^{+0T}$,
while the magic state $\ket{H}$ belongs to $\mathcal{C}_{1}^{+1T}$.

A bit tricky point is that, even if $U$ is a \emph{strict} Clifford+$kT$ gate,
$U\ket{0}^{\otimes n}$ is not guaranteed to be a \emph{strict} Clifford+$kT$ state,
as can be seen from a simple example of $T\ket{0} = \ket{0}$.
To examine this issue more carefully,
let $\ket{\psi} = U\ket{0}^{\otimes n}$ be a strict Clifford+$k'T$ state (with $k' \leq k$ by definition) which can be written as $U'\ket{0}^{\otimes n}$
for some strict Clifford+$k'T$ gate $U'$.
For generic $n$, the condition which must hold is 
\begin{equation*}
	U'^\dagger U\ket{0}^{\otimes n} = \ket{0}^{\otimes n}.
\end{equation*}
In particular, for the simplest case of $n=1$ this is equivalent to
\begin{equation*}
	U'^\dagger U
	=
	\left(
 		\begin{array}{wc{16pt}wc{20pt}}
 			1 & 0\\
 			0 & e^{i\phi}
 		\end{array}
 	\right)
\end{equation*}
for some $\phi \in \mathbb{R}$.
Fortunately, it is known that the possible value of $\phi$ is limited to integer multiples of $\frac{\pi}{4}$~\cite[Appendix A]{KliuchnikovMaslovMosca:2012},
meaning that $U'^\dagger U = T^m$ $(m \in \mathbb{Z}/8\mathbb{Z})$.

Based on Definition~\ref{def:Clifford+kT_states},
one can define the notion of Clifford$+\Delta kT$ channels:
\begin{definition}[Clifford$+\Delta kT$ channels]
	A quantum channel $\Phi$ is called a Clifford$+\Delta kT$ channel
	if it sends any (pure) strict Clifford+$kT$ state to a (in general mixed and non-strict) Clifford$+(k+\Delta k)T$ state.
\end{definition}
\noindent
Note that (as always) the input and output states need not be those of the same number of qubits. 

\newpage

\section{Clifford+$kT$ robustness} \label{sec:CkT robustness}

We now introduce the concept of Clifford+$kT$ robustness and discuss its properties.
This section constitutes our theoretical contribution of this study.

\subsection{Definition}
A quantum state with any kind of \emph{resource} typically loses its resource when it is mixed with a \emph{free} (\textit{i.e.} null-resource) state.
This motivates one to define \emph{robustness} of a state~$\rho$ 
as the minimum amount of free-state mixing to completely destroy the resource of $\rho$.
This can be formulated mathematically as
\begin{equation*}
	\min_{\sigma \in \mathcal{F}}
	\left\{
		t\in \mathbb{R}_{\geq 0} \,\bigg|\,\frac{\rho + t\sigma}{1+t} \in \mathcal{F}
	\right\}
\end{equation*}
where $\mathcal{F}$ is the set of all possible free states.
It turns out that this quantity serves as a useful measure in various resource theories~\cite{VidalTarrach:1998,chitambar2019}.

If free states with respect to the resource under consideration form a convex set
(with large enough dimension),
the robustness in fact has a simple expression as follows.
First, given such an $\mathcal{F}$, a generic state $\rho$ can always (not uniquely) be expanded into a \emph{pseudo-mixture} of free states $\{\sigma_i\}$ as
\begin{equation*}
	\rho = \sum_i
	c_i\sigma_i
	\quad
	\bigg(\sum_i c_i = 1\bigg).
\end{equation*}
Here, the prefix \emph{pseudo-} indicates that some of the $c_i$'s might be negative.
Rewriting this as
\begin{equation*}
	\rho +
	\sum_{c_i < 0}
	(-c_i) \sigma_i
	=
	\sum_{c_i > 0}
	c_i\sigma_i
\end{equation*}
where both $\sum_{c_i < 0} (-c_i) \sigma_i$ and $\sum_{c_i > 0}c_i\sigma_i$ are (non-normalized) mixed free states by definition,
the robustness of $\rho$ reduces to
\begin{equation*}
	\min_{\substack{\{c_i\} \text{ s.t.} \\ 
    \rho = \sum_i c_i \sigma_i}} 
	\left( \sum_{c_i < 0}(-c_i) \right).
\end{equation*}
Note that it is often useful to consider $\min_{\{c_i\}} \sum_i |c_i|$ instead, since
\begin{equation} \label{eq:1+2R}
	\begin{array}{lcl}
		\displaystyle
		\sum_i |c_i|
		& = &
		\displaystyle
		\sum_{c_i > 0}c_i
		+
		\sum_{c_i < 0}(-c_i)\\[16pt]
		& = &
		\displaystyle
		\underbrace{\sum_i c_i}_{=1}
		+
		2\sum_{c_i < 0}(-c_i)\\[16pt]
	\end{array}
\end{equation}
and thus it contains the same information as the original robustness.
Note that $\min_{\{c_i\}} \sum_i |c_i| \geq 1$ always holds.

Now, let us introduce the concept of \emph{Clifford+$kT$ robustness}, which quantifies the amount of non-Clifford-ness of a given quantum state.

\begin{definition}[Clifford+$kT$ robustness]\label{def:Clifford+kT_robustness}
	The Clifford+$kT$ robustness of a quantum state $\rho$ is defined as
	\begin{equation*}\label{eq:definition_robustness}
		R_k(\rho) = \min_{\{c_i\}}
		\bigg\{
			\sum_i |c_i|
			\ \bigg|\ 
			\rho = \sum_i c_i \ket{\psi_i}\!\bra{\psi_i}
			\ \Big(\ket{\psi_i} \in \mathcal{C}_{n}^{+(\leq k)T}, \sum_i c_i = 1\Big)
		\bigg\}
	\end{equation*}
	minimized over all possible decomposition of $\rho$ in terms of (non-strict) Clifford+$kT$ pseudo-mixture.
	We occasionally abuse the notation and let $R_k(\ket{\psi})$ denote $R_k(\ket{\psi}\!\bra{\psi})$ for pure states.
\end{definition}

\noindent
This is in fact a rather straightforward generalization (and can be regarded as a finer version) of the notion of \emph{robustness of magic (RoM)} \cite{HowardCampbell:2016} (which corresponds to the $k=0$ case).
In the following,
we discuss the implications of this measure for quantum computation.


\subsection{Basic properties as a resource measure}
First of all, the Clifford+$kT$ robustness defined above distinguishes Clifford+$kT$ (free) states from the other (resourceful) states:

\begin{theorem}[Faithful]\label{prop:faithful}
	For any quantum state $\rho$,
	$R_k(\rho) = 1$
	if and only if $\rho$ is a (possibly mixed) Clifford+$kT$ state,
	and $R_k(\rho) > 1$ otherwise.
\end{theorem}

\begin{proof}
 $R_k(\rho) > 1$ holds if and only if there are coefficients satisfying $c_i < 0$ in Eq.~\eqref{eq:1+2R}.
\end{proof}

\noindent
Furthermore, it captures the reduction of resource under free-state mixing in the following sense:
\begin{theorem}[Monotone]\label{prop:monotone}
	For any quantum state $\rho$ and any trace-preserving Clifford$+\Delta kT$ channel $\Phi$,
	\begin{equation*}
		R_{k+\Delta k}(\Phi(\rho))
		\leq
		R_k(\rho).
	\end{equation*}
\end{theorem}
	
\begin{proof}
	Let $\Phi$ be such that $\Phi(\ket{\psi_i}\!\bra{\psi_i}) = \sum_{\ket{\phi_j}\,\in\,\mathcal{C}_{n'}^{+(\leq k+\Delta k)T}} \tilde{c}_{ij} \ket{\phi_j}\!\bra{\phi_j}$
	for a strict Clifford+$kT$ state $\ket{\psi_i} \in \mathcal{C}_{n}^{+kT}$.
	Here, $\tilde{c}_{ij} \geq 0$ and the trace-preserving condition requires $\sum_j \tilde{c}_{ij} = 1$.
	For the Clifford+$kT$ pseudo-mixture expansion
	$\rho = \sum_{\ket{\psi_i}\,\in\,\mathcal{C}_{n}^{+(\leq k)T}} c_i \ket{\psi_i}\!\bra{\psi_i}$
	attaining the minimum in the definition of Clifford+$kT$ robustness, one has
	\begin{align*}
		R_{k+\Delta k}(\Phi(\rho))
		& =
		R_{k+\Delta k}\Big(
			\sum_{\ket{\psi_i}\,\in\,\mathcal{C}_{n}^{+(\leq k)T}} c_i
			\sum_{\ket{\phi_j}\,\in\,\mathcal{C}_{n'}^{+(\leq k+\Delta k)T}}
			\tilde{c}_{ij} \ket{\phi_j}\!\bra{\phi_j}
		\Big)\\
		& \leq \sum_j \Big| \sum_i c_i\,\tilde{c}_{ij} \,\Big|\\
		& \leq \sum_i |c_i| \underbrace{
			\sum_j |\tilde{c}_{ij}|
		}_{=1}. \qedhere
	\end{align*}
\end{proof}

In addition to these two fundamental properties, the Clifford+$kT$ robustness has several nice properties:

\begin{theorem}[Convex]\label{prop:convex}
	For any quantum states $\{\rho_i\}$,
	\begin{equation*}
		R_k\big(\sum_i p_i \rho_i\big)
		\leq
		\sum_i |p_i| \cdot R_k(\rho_i).
	\end{equation*}
\end{theorem}
	
\begin{proof}
	For Clifford+$kT$ pseudo-mixture expansions
	$
		\rho_i
		=
		\sum_{\ket{\psi_{i,j}}\,\in\,\mathcal{C}_{n}^{+(\leq k)T}} c_{i,j} \ket{\psi_{i,j}}\!\bra{\psi_{i,j}}
	$
	attaining the minimum in the definition of Clifford+$kT$ robustness, one has
	\begin{align*}
		R_k\big(\sum_i p_i \rho_i\big)
		& = R_k\big(\sum_i\sum_j p_i c_{i,j} \ket{\psi_{i,j}}\!\bra{\psi_{i,j}} \big)\\
		& \leq \sum_i\sum_j |p_i c_{i,j}|\\
		& = \sum_i \sum_j |p_i|\cdot |c_{i,j}|\\
		& = \sum_i  |p_i| \cdot \sum_j |c_{i,j}|.\qedhere
	\end{align*}
\end{proof}

\begin{theorem}[Sub-multiplicative]\label{prop:sub-multiplicative}
	For any two quantum states $\rho$ and $\rho'$,
	\begin{equation*}
		R_{k+k'}(\rho \otimes \rho')
		\leq
		R_{k}(\rho)\cdot R_{k'}(\rho').
	\end{equation*}
\end{theorem}
	
\begin{proof}
	For a Clifford+$kT$ pseudo-mixture expansion
	$\rho = \sum_{\ket{\psi_i}\,\in\,\mathcal{C}_{n}^{+(\leq k)T}} c_i \ket{\psi_i}\!\bra{\psi_i}$
	and a Clifford+$k'T$ pseudo-mixture expansion
	$\rho' = \sum_{\ket{\phi_j}\,\in\,\mathcal{C}_{n'}^{+(\leq k')T}} c'_j \ket{\phi_j}\!\bra{\phi_j}$
	both attaining the minimum in the definition of Clifford+$kT$ robustness,
	the state is given by
	\begin{equation*}
		\rho\otimes \rho'
		=
		\sum_i \sum_j c_i c'_j
		\ket{\psi_i}\!\bra{\psi_i}
		\otimes
		\ket{\phi_j}\!\bra{\phi_j},
	\end{equation*}
	and since 
	$\ket{\psi_i}
	\otimes
	\ket{\phi_j} \in \mathcal{C}_{n+n'}^{+(\leq k+k')T}$,
	the Clifford$+(k+k')T$ robustness of $\rho\otimes \rho'$ is upper-bounded as
	\begin{align*}
		R_{k+k'}(\rho\otimes \rho')
		& \leq
		\sum_i \sum_j |c_i c'_j|\\
		& =
		\sum_i \sum_j |c_i| \cdot |c'_j|\\
		& =
		\sum_i  |c_i| \cdot \sum_j |c'_j|.\qedhere
	\end{align*}
\end{proof}


\subsection{Lower bound} \label{subsec:lower bound}
Here we derive a lower bound of Clifford+$kT$ robustness.
For an $n$-qubit quantum state~$\rho$ and its Clifford+$kT$ pseudo-mixture expansion
$\rho= \sum_{\ket{\psi_i}\,\in\,\mathcal{C}_{n}^{+(\leq k)T}} c_i \ket{\psi_i}\!\bra{\psi_i}$, one has
\begin{equation*}
	\Tr(P_a\rho)
	=
	\sum_i c_i
	\Tr(P_a\ket{\psi_i}\!\bra{\psi_i})
\end{equation*}
where $P_a$ $(a=1,\cdots,4^n)$ is an $n$-qubit Pauli operator.
Letting $A_{ai} = \Tr(P_a\ket{\psi_i}\!\bra{\psi_i})$, $x_i = c_i$, and $b_a = \Tr(P_a\rho)$,
the minimization in Definition \ref{def:Clifford+kT_robustness} of Clifford+$kT$ robustness can be rewritten as a basis pursuit problem
\begin{equation*}
	\min_{\bm{x}}
	\Big\{
		\|\bm{x}\|_1
		\, \Big|\, 
		A\bm{x} = \bm{b}
	\Big\}.
\end{equation*}
As is well known, this can be turned into a linear program (LP) by further decomposing $\bm{x} = \bm{x}^+ - \bm{x}^-$
such that $\bm{x}^+, \bm{x}^- \geq \bm{0}$ as
\begin{equation*}
	\min_{
		{\tiny
			\begin{pmatrix}
				\bm{x}^+\\
				\bm{x}^-\\
			\end{pmatrix}
		}
		\geq
		\bm{0}
	}
	\Big\{
		\left(
			\begin{array}{ccc}
				1 & \cdots & 1
			\end{array}
		\right)
		\left(
			\begin{array}{c}
				\bm{x}^+\\
				\bm{x}^-\\
			\end{array}
		\right)
		\, \Big|\, 
		\left(
			\begin{array}{cc}
				A & \ -A
			\end{array}
		\right)
		\left(
			\begin{array}{c}
				\bm{x}^+\\
				\bm{x}^-\\
			\end{array}
		\right) = \bm{b}
	\Big\}.
\end{equation*}
The dual LP is given by
\begin{equation*}
	\max_{\bm{y}}
	\Big\{
		\bm{b}^{\mathsf{T}}\bm{y}
		\,\Big|\,
		\left(
			\begin{array}{c}
				A^{\mathsf{T}}\\
				-A^{\mathsf{T}}
			\end{array}
		\right)\bm{y}
		\leq
		\left(
			\begin{array}{c}
				1\\
				\vdots\\
				1
			\end{array}
		\right)
	\Big\}
\end{equation*}
and writing $Y = \sum_a y_a P_a$ this is further equivalent to
\begin{equation}\label{eq:basis pursuit}
	\max_{Y}
	\Big\{
		\Tr(Y\rho)
		\,\Big|\,
		-\!1 \leq \Tr(Y\ket{\psi_i}\!\bra{\psi_i}) \leq 1
		\ \ 
		\forall\ket{\psi_i}\in \mathcal{C}_{n}^{+(\leq k)T}
	\Big\}.
\end{equation}

\newpage

In fact, one can use this to derive a lower bound of the Clifford+$kT$ robustness $R_k(\rho)$
as follows.
Let us begin with showing the following lemma:
\begin{lemma}
 For any Clifford+$kT$ state $\ket{\psi_i}$,
	\begin{equation*}
		\sum_{a=1}^{4^n} \big|\Tr(P_a\ket{\psi_i}\!\bra{\psi_i}) \big|
		\leq
		2^n (\sqrt{2})^k.
	\end{equation*}
\end{lemma}

\begin{proof}
	Given a decomposition of a $T$-gate channel into a mixed-Clifford channel for a state~$\sigma$ as
	\begin{equation*}
		T \sigma T^\dagger
		=
		\sum_{l} d_l\, D_l \,\sigma D_l^\dagger,
	\end{equation*}
    where $d_l$'s are coefficients and $D_l$'s are Clifford gates,
	a Clifford+$kT$ state $\ket{\psi_i}$ is decomposed (see Definitions \ref{def:Clifford+kT_gates} and \ref{def:Clifford+kT_states}) as
	\begin{align*}
		\ket{\psi_i}\!\bra{\psi_i}
		& =
		C_kT^{(q_k)} \cdots C_1 T^{(q_1)}C_0\ket{0}^{\otimes n}
		\bra{0}^{\otimes n}
		C_0^\dagger
		T^{(q_1)\dagger}
		C_1^\dagger
		\cdots
		T^{(q_k)\dagger}
		C_k^\dagger\\
		& =
		C_kT^{(q_k)} \cdots C_1 \bigg(
			\sum_{l_1} d_{l_1} D_{l_1} C_0\ket{0}^{\otimes n}
			\bra{0}^{\otimes n}
			C_0^\dagger
			D_{l_1}^{\dagger}
		\bigg)
		C_1^\dagger
		\cdots
		T^{(q_k)\dagger}
		C_k^\dagger\\[-6pt]
		& = \cdots\\[4pt]
		& = \sum_{l_1,\cdots, l_k} d_{l_1} \cdots d_{l_k}
		\underbrace{C_k D_{l_k} \cdots C_1 D_{l_1} C_0}_{\eqqcolon C_{l_1,\cdots, l_k}}
		\ket{0}^{\otimes n}\bra{0}^{\otimes n}
		\underbrace{C_0^\dagger D_{l_1}^{\dagger} C_1^\dagger \cdots D_{l_k}^{\dagger} C_k^\dagger}_{= C_{l_1,\cdots, l_k}^\dagger}
	\end{align*}
	and
	\begin{equation*}
		\sum_{a=1}^{4^n} \big|\Tr(P_a\ket{\psi_i}\!\bra{\psi_i}) \big|
		=
		\sum_{a=1}^{4^n}
		\Big|
			\sum_{l_1,\cdots,l_k} d_{l_1}\cdots d_{l_k}
			\Tr(
				C_{l_1,\cdots, l_k}^\dagger P_a C_{l_1,\cdots, l_k}
				\ket{0}^{\otimes n}\!\bra{0}^{\otimes n}
			)
		\Big|.
	\end{equation*}
	Noting that the trace on the right-hand side is $1$ if $C_{\{l_i\}}^\dagger P_a C_{\{l_i\}} \in \{I,Z\}^{\otimes n}$ and otherwise $0$,
	\begin{align*}
    \sum_{a=1}^{4^n} \big|\Tr(P_a\ket{\psi_i}\!\bra{\psi_i}) \big|
		& \leq
		\Big|
			\sum_{l_1,\cdots,l_k} d_{l_1}\cdots d_{l_k}
		\Big|
		\cdot 2^n\\
		& \leq
		\bigg(\sum_l |d_l|\bigg)^k \cdot 2^n.
	\end{align*}
	Here, a decomposition of the $T$-gate channel
	\begin{equation*}
		T\sigma T^\dagger = \frac{1}{2}I\sigma I^\dagger + \frac{1}{\sqrt{2}}S\sigma S^\dagger + \left(\frac{1}{2}-\frac{1}{\sqrt{2}}\right)Z \sigma Z^\dagger
	\end{equation*}
	gives $\sum_l |d_l| = \frac{1}{2} + \frac{1}{\sqrt{2}} + (\frac{1}{\sqrt{2}} - \frac{1}{2}) = \sqrt{2}$, and hence the lemma follows.
\end{proof}


Given the lemma,
a guess $y_a = \frac{1}{2^n(\sqrt{2})^k}\mathrm{sgn}(\Tr(P_a\rho))$ as in \cite[Appendix B]{HeinrichGross:2018} is feasible as
\begin{equation*}
	\begin{array}{lcl}
		\big|\Tr (Y\ket{\psi_i}\!\bra{\psi_i})\big|
		& = &
		\displaystyle
		\frac{1}{2^n(\sqrt{2})^k}\sum_a \big|\mathrm{sgn}(\Tr(P_a\rho)) \Tr(P_a\ket{\psi_i}\!\bra{\psi_i})\big|\\
		& \leq & 
		\displaystyle
		\frac{1}{2^n(\sqrt{2})^k}\sum_a \big|\Tr(P_a\ket{\psi_i}\!\bra{\psi_i})\big|\\
		& \leq &
		1,\\
	\end{array}
\end{equation*}
and the corresponding value of $\Tr (Y\rho)$ is
\begin{align*}
	\Tr (Y\rho)
	& = \Tr \bigg(
		\sum_a
		\frac{1}{2^n(\sqrt{2})^k}\mathrm{sgn}(\Tr(P_a\rho))
		P_a\rho
	\bigg)\\
	& = \frac{1}{2^n(\sqrt{2})^k} \sum_a \bigg(
		\mathrm{sgn}(\Tr(P_a\rho))
		\Tr (P_a\rho)
	\bigg).
\end{align*}

\noindent
This gives a lower bound of Clifford+$kT$ robustness:
\begin{theorem}[A lower bound of Clifford+$kT$ robustness]\label{theorem:lower_bound}
	\begin{equation*}
		R_k(\rho) \geq \frac{1}{2^n(\sqrt{2})^k}
		\sum_{a=1}^{4^n} \big|\Tr(P_a \rho)\big|.
	\end{equation*}
\end{theorem}
One interesting implication of this lower bound can be found by considering $n$-qubit tensor products of the magic states, $\ket{H}^{\otimes n}$ and $\ket{SH}^{\otimes n}$ [Eqs.~\eqref{eq:def of H state} and \eqref{eq:def of SH state}].
One has
\begin{equation*}
	\begin{array}{lclclcl}
		\ket{H}^{\otimes n} & : &
		\displaystyle
		\sum_a \big|\Tr(P_a\rho)\big|
		& = &
		(1+\sqrt{2})^n\\
		\ket{SH}^{\otimes n} & : &
		\displaystyle
		\sum_a \big|\Tr(P_a\rho)\big|
		& = &
		(1+\sqrt{3})^n\\
	\end{array}
\end{equation*}
and in the limit of $n,k \to\infty$ with fixed $\frac{k}{n}$, Clifford+$kT$ robustness $R_k(\rho)$ is guaranteed to grow (exponentially in $n$) for
\begin{equation*}
	\begin{array}{lclclcl}
		\ket{H}^{\otimes n} & : &
		\frac{k}{n}
		& < &
		2\log_2 \frac{1+\sqrt{2}}{2} \approx 0.54311,\\
		\ket{SH}^{\otimes n} & : &
		\frac{k}{n}
		& < &
		2\log_2 \frac{1+\sqrt{3}}{2} \approx 0.89997.\\
	\end{array}
\end{equation*}
Combining this with the relationship between $R_k(\rho)$ and the sampling cost of simulating a state $\rho$ using Clifford+$kT$ states (to be discussed in Sec.~\ref{subsec:sampling overhead}), we can guess, for example, that a quantum circuit with $n$ $T$ gates cannot be efficiently simulated in a way described in Sec.~\ref{subsec:sampling overhead} when $\frac{k}{n} \lesssim 0.54311$ for large $n$ and $k$. 

\subsection{Upper bound} \label{subsec:upper bound}
Since an identity operation is trivially (non-strict) Clifford+$kT$ for any $k$,
Theorem~\ref{prop:monotone} implies that the Clifford+$kT$ robustness of a given state is monotonic in $k$:
\begin{corollary}\label{corollary:monotone_in_k}
	For any quantum state $\rho$,
	\begin{equation*}
		R_0(\rho) \geq R_1(\rho) \geq \cdots \geq R_k(\rho) \geq R_{k+1}(\rho) \geq \cdots.
	\end{equation*}
\end{corollary}
\noindent
This gives a rather trivial (\textit{i.e.} $k$-independent) upper bound on Clifford+$kT$ robustness, which is namely RoM (or its upper bound).
However, the existing analyses of upper bounds on RoM (\textit{e.g.}~\cite{LiuWinter:2020}) typically rely on special properties of stabilizer states,
which hinder their naive extension to the Clifford+$kT$ robustness case.
We leave the derivation of a non-trivial ($k$-dependent) upper bound for future work.

\subsection{Relation to simulation overhead} \label{subsec:sampling overhead}
As in the case of RoM, one can show that Clifford+$kT$ robustness of a given quantum state is related to its simulation cost  by sampling Clifford+$kT$ states. 
Consider a task of simulating generic (in particular highly non-Clifford) quantum circuits
while only allowed to (efficiently) implement Clifford+$kT$ circuits on an early-FTQC device.
More concretely, the task is to estimate an expectation value $\Tr (P\rho)$ of a Pauli operator $P$
for a quantum state $\rho$ obtained by applying the circuit to an input state $\ket{0}^{\otimes n}$.

Given a decomposition of $\rho$ into a Clifford+$kT$ pseudo-mixture (which need not be optimal), the estimand is
\begin{equation*}
	\begin{array}{lcl}
		\Tr (P\rho)
		& = &
		\displaystyle
		\sum_i c_i \Tr(P\ket{\psi_i}\!\bra{\psi_i})\\
		& = &
		\displaystyle
		\sum_i \frac{|c_i|}{\sum_j |c_j|} \cdot \bigg[
			\underbrace{
				\frac{c_i}{|c_i|}
			}_{
				=\mathrm{sgn}\,c_i
			} \cdot \bigg(\sum_j |c_j|\bigg) \cdot \Tr(P\ket{\psi_i}\!\bra{\psi_i})
		\bigg]
	\end{array}
\end{equation*}
and can be interpreted as an expectation value $\mathbb{E}[X]$ of a random variable $X$ where
$X(i) = \left( \mathrm{sgn}\,c_i \right) \cdot \big(\sum_j |c_j|\big) \cdot \Tr(P\ket{\psi_i}\!\bra{\psi_i})$
with a probability function $P(\{i\}) = \frac{|c_i|}{\sum_j |c_j|}$.
Therefore, by repeatedly sampling $i$ and computing $X(i)$ by evaluating the expectation values $\Tr (P\ket{\psi_i}\!\bra{\psi_i})$ for Clifford+$kT$ states (using either early-FTQC devices or classical computers), we can obtain an estimate of the target quantity.
Since the estimand satisfies $|X(i)|\leq \big(\sum_j |c_j|\big)$, the Hoeffding's inequality \cite[Theorem 2]{Hoeffding:1963} tells us
\begin{equation*}
	P\bigg(
		\Big|\frac{1}{N}\sum_{n=1}^N X^{(n)} -\Tr (P \rho)\Big| \geq \delta
	\bigg)
	\leq
	2\exp\bigg(-\frac{2N^2\delta^2}{N (2\sum_j |c_j|)^2}\bigg).
\end{equation*}
From this one can deduce the number of samples $N$ needed to estimate $\Tr (P \rho)$ within error~$\delta$ with probability~$1-\epsilon$ as
\begin{equation*}
	N
	\geq
	\frac{2}{\delta^2}
	\Big(
		\underbrace{
			\sum_j |c_j|
		}_{
			\geq R_k(\rho)
		}
	\Big)^2
	\log_e \frac{2}{\epsilon},
\end{equation*}
which exhibits quadratic growth in Clifford+$kT$ robustness for the optimal decomposition.

A notable application of the above argument is to the case where one equipped with Clifford+$kT$ computational power wants to simulate an $n$-qubit state $\rho$
but $n$ is so large that one is ignorant of its optimal decomposition into Clifford+$kT$ pseudo-mixture.
If $\rho$ (is not completely mixed and) admits a tensor product representation
\begin{equation*}
	\rho = \bigotimes_{i} \rho_i
\end{equation*}
together with a weak composition $(k_1, k_2, \cdots)$ of $k$ such that $k = \sum_i k_i$ $(k_i \geq 0)$
and optimal decompositions of $\rho_i$ into Clifford+$k_iT$ pseudo-mixture is known for each $i$,
then one can accomplish the task with cost
\begin{equation*}
	\sum_j |c_j|
	=
	\prod_i R_{k_i}(\rho_i)
\end{equation*}
instead of (smaller) $R_k(\rho)$. 
We will see concrete examples in Sec.~\ref{subsec:numerical results of robustness}.

\subsection{Operational implications} \label{subsec:operational meaning}
\label{subsec:convertibility} 
The monotonicity (Theorem~\ref{prop:monotone}) of Clifford+$kT$ robustness
also allows us to examine whether two quantum states are convertible under Clifford+$kT$ channels:
\begin{corollary}[Inconvertibility] \label{corollary:convertibility}
	For any two quantum states $\rho$ and $\rho'$ satisfying
	\begin{equation*}
		R_{k+\Delta k}(\rho') > R_k(\rho),
	\end{equation*}
	it is impossible to convert $\rho$ into $\rho'$ by a Clifford$+\Delta kT$ channel.
\end{corollary}
\begin{proof}
If there exists a Clifford+$\Delta kT$ channel $\Phi$ converting $\rho$ into $\rho' = \Phi(\rho)$,
then the monotonicity implies $R_{k+\Delta k}(\rho') = R_{k+\Delta k}(\Phi(\rho)) \leq R_k(\rho)$,
which contradicts the assumption. 
\end{proof}

This point of view is especially useful for investigating problems in gate synthesis:
\begin{corollary}[Minimum $T$-gate requirement]\label{corollary:gate_synthesis_no_go}
	For any $n$-qubit gate~$U$ such that
	\begin{equation*}
		R_{k+\Delta k}(U\ket{+}^{\otimes n}) > R_k((T\ket{+})^{\otimes {(k' - \Delta k)}})
	\end{equation*}
	for some $k$ and $0 \leq \Delta k \leq k'$,
	it is impossible to synthesize $U$ using only $k'$ $T$ gates.
\end{corollary}
\begin{proof}
	It suffices to prove the existence of an $n$-qubit state $\ket{\psi}$ such that
	$\ket{\psi}\otimes (T\ket{+})^{\otimes (k' - \Delta k)}$ cannot be converted to $U\ket{\psi}$ under Clifford$+\Delta kT$ channels, given the assumption.
	One can see that $\ket{\psi} = \ket{+}^{\otimes n}$ is exactly such a state by
	applying Corollary~\ref{corollary:convertibility} to
	\begin{itemize}
	\item $\rho$ corresponding to $\ket{+}^{\otimes n} \otimes (T\ket{+})^{\otimes (k'-\Delta k)}$ and
	\item $\rho'$ corresponding to $U\ket{+}^{\otimes n}$
	\end{itemize}
	as the sub-multiplicativity of Clifford+$kT$ robustness (Theorem~\ref{prop:sub-multiplicative}) leads to
	\begin{equation*}
		R_{k}(\rho) \leq \underbrace{R_0(\ket{+}^{\otimes n})}_{=1} \cdot R_k((T\ket{+})^{\otimes (k' - \Delta k)}) < R_{k+\Delta k}(U\ket{+}^{\otimes n}) = R_{k+\Delta k}(\rho'). \qedhere
	\end{equation*}
\end{proof}

\noindent
Note that the $k=0, \Delta k = 0$ case of Corollary~\ref{corollary:gate_synthesis_no_go} reproduces the known result for RoM.
Conversely, Clifford+$kT$ robustness gives more precise information on the requirement than RoM
in certain cases \textit{e.g.} when $U$ satisfies
\begin{equation*}
	\begin{array}{ccccccc}
		R_0(U\ket{+}^{\otimes n}) & \geq & \cdots & \geq & R_k(U\ket{+}^{\otimes n}) & \geq & \cdots\\
		\rotatebox[]{90}{$\geq$} &&&& \rotatebox[]{90}{$<$}\\
		R_0((T\ket{+})^{\otimes k'}) & \geq & \cdots & \geq & R_k((T\ket{+})^{\otimes k'}) & \geq & \cdots
	\end{array}
\end{equation*}
for some $k$.

\newpage

\section{Numerical observations \label{sec:numerics}}
Here we display various results obtained by numerical experiments as well as some additional theoretical analysis inspired by them.

\subsection{Enumeration of Clifford+$kT$ states} \label{subsec:count of Clifford+kT}

First of all, the number of $n$-qubit Clifford$+0T$ (\textit{i.e.} stabilizer) states is known \cite{AaronsonGottesman:2004} to be
\begin{equation} \label{eq:Clifford+0T_state_count}
	2^n \prod_{j=0}^{n-1} (2^{n-j} + 1).
\end{equation}
Since it is also known (\textit{e.g.}~\cite{CalderbankRainsShorSloane:1996}) that there are
\begin{equation*}
	2^{n^2+2n} \prod_{j=1}^n (4^j - 1)
\end{equation*}
Clifford gates in total (up to global phase factors), a brute-force search through $2^{\mathcal{O}(n^2 k)}$ states in principle allows us to enumerate all the Clifford+$kT$ states.
But one carry out the task more efficiently by employing the following fact (based on Eq.\,\eqref{eq:Clifford+kT_gate_seq_rot_normal_form}):

\begin{proposition}[{\cite[Proposition 1]{GossetKliuchnikovMoscaRusso:2013}}]\label{prop:n-qubit_strict_Clifford+kT_state_normal_form}
	Any strict Clifford+$kT$ state can be constructed (up to a global phase factor) by applying $k$ sequential non-trivial Pauli rotations $R_{P_i}(+\frac{\pi}{4})$ $(i=1,\cdots, k)$ to a stabilizer state.
\end{proposition}

\noindent
As there are (only) $4^n$ Pauli operators, the number of states to be searched through can be reduced to $2^{\mathcal{O}(nk)}$.
The numerical results are summarized in the following tables and figure.\footnote{
	\setlength{\baselineskip}{16pt}%
	We used \texttt{Qulacs}~\cite{suzuki2021qulacsfast} to generate candidate state vectors.
	Each newly generated state was compared with previously generated ones, after rounding its amplitudes (elements) to ten digits of precision.
}

It seems that the number of distinct $n$-qubit Clifford+$kT$ states for fixed $n$ grows exponentially in $k$ (the left panel of Fig.~\ref{fig:num_states}).
This results in a huge number of states even for small $n$ and $k$, making further numerical simulation impractical.
The scaling in $n$ for fixed $k$ (the right panel of Fig.~\ref{fig:num_states}) is not clear given our limited data.
It appears to be a simple exponential growth at least for small $n$, but the results for $k=0$ show deviation for larger $n$,
although it (still) does not match the scaling $2^{\Theta(n^2)}$ expected from Eq.~\eqref{eq:Clifford+0T_state_count}.

\newpage

\begin{table*}[!h]
		\begin{tabular}{wc{40pt}|wc{40pt}wc{40pt}wc{40pt}wc{40pt}wc{40pt}wc{40pt}wc{32pt}wc{32pt}wc{32pt}wc{32pt}wc{32pt}}
			$k$ & $0$ & $1$ & $2$ & $3$ & $4$ & $5$ & $6$ & $7$ & $8$ & $9$ & $10$ \\
			\hline
			$n=1$ & $6$			 & $18$				& $42$			& $90$			  & $186$		 & $378$		  & $762$	  & $1530$ & $3066$ & $6138$ & $12282$ \\
			$n=2$ & $60$		& $420$			 & $2580$	   & $18900$	 & $134100$ & $1040340$ & $\cdots$ &				&				&				 & \\
			$n=3$ & $1080$	 & $16200$		& $227880$ & $4098600$ & $\cdots$	&						&				  &				   &				&				& \\
			$n=4$ & $36720$ & $1138320$ & $\cdots$	  &						 &					  &						 &					&				&				 &				 & \\
			$\vdots$
	\end{tabular}
	\caption{The number of distinct $n$-qubit (non-strict) Clifford+$kT$ states.}
\end{table*}

\begin{table*}[!h]
		\begin{tabular}{wc{40pt}|wc{40pt}wc{40pt}wc{40pt}wc{40pt}wc{40pt}wc{40pt}wc{32pt}wc{32pt}wc{32pt}wc{32pt}wc{32pt}}
			$k$ & $0$ & $1$ & $2$ & $3$ & $4$ & $5$ & $6$ & $7$ & $8$ & $9$ & $10$ \\
			\hline
			$n=1$ & $6$			 & $12$				& $24$			& $48$			  & $96$		 & $192$		  & $384$	  & $768$ & $1536$ & $3072$ & $6144$ \\
			$n=2$ & $60$		& $360$			 & $2160$	   & $16320$	 & $115200$ & $906240$ & $\cdots$ &				&				&				 & \\
			$n=3$ & $1080$	 & $15120$		& $211680$ & $3870720$ & $\cdots$	&						&				  &				   &				&				& \\
			$n=4$ & $36720$ & $1101600$ & $\cdots$	  &						 &					  &						 &					&				&				 &				 & \\
			$\vdots$
	\end{tabular}
	\caption{The number of distinct $n$-qubit strict Clifford+$kT$ states.}
	\label{tab:number_of_1-qubit_strict_Clifford+kT_states}
\end{table*}

\begin{figure}[!h]
	\begin{minipage}{0.48\textwidth}
		\centering
		\includegraphics[width=\textwidth]{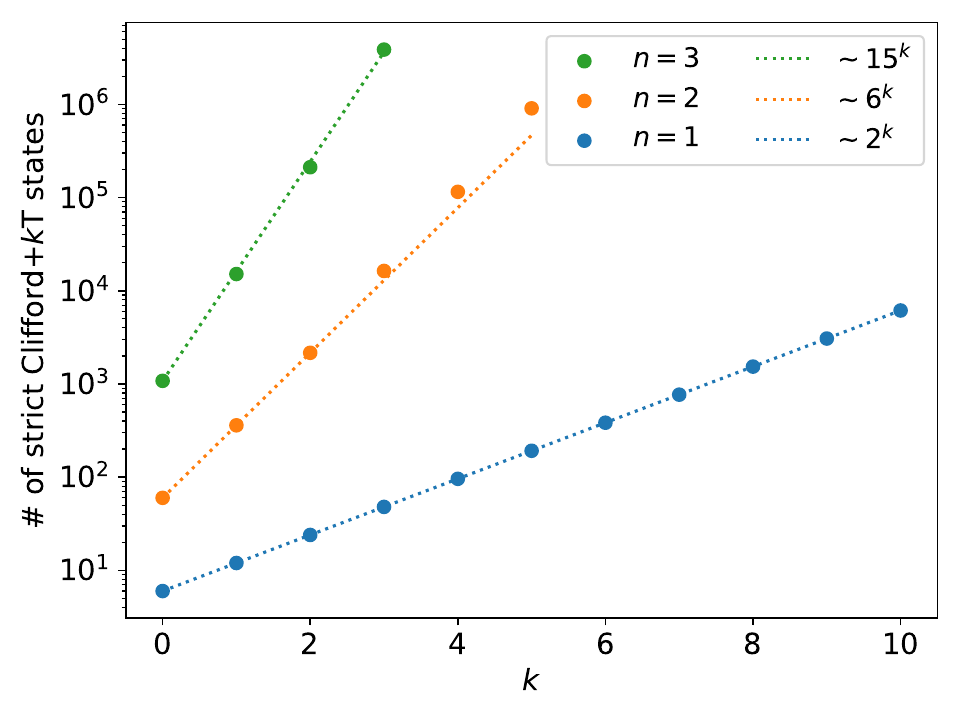}
	\end{minipage}
	\begin{minipage}{0.48\textwidth}
		\centering
		\includegraphics[width=\textwidth]{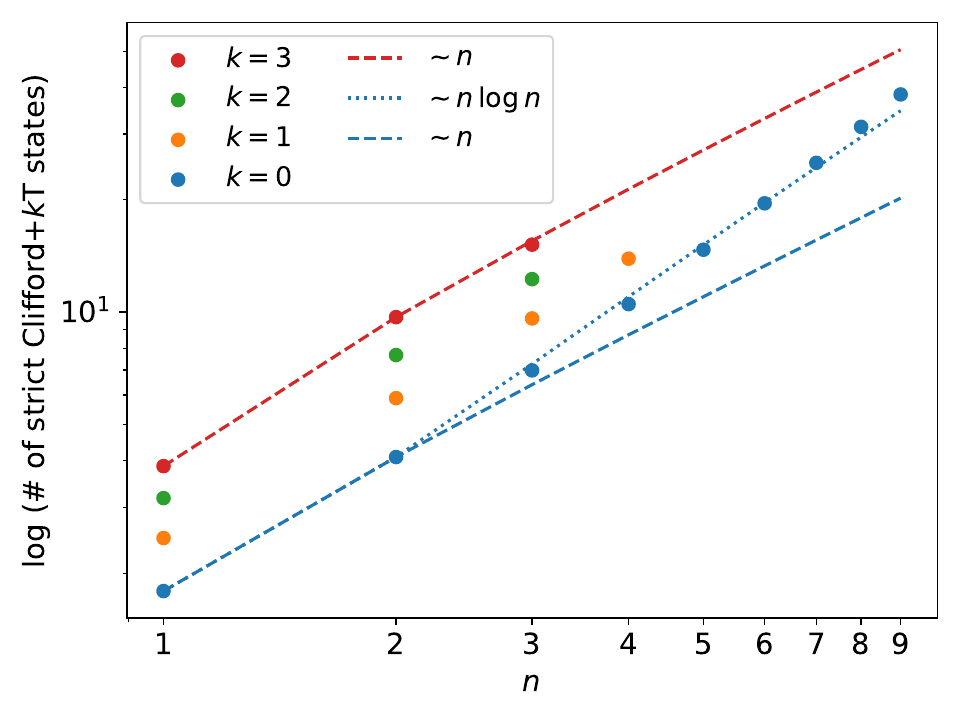}
	\end{minipage}
	\caption{
    Plots of the number of $n$-qubit distinct strict Clifford+$kT$ states (from Table~\ref{tab:number_of_1-qubit_strict_Clifford+kT_states}).
    In the right panel, the lines are drawn by fitting the parameters of $y = ax + b$ (resp. $y = a'x\log x +b'$) using the first two ($n=1,2$) data points.
    }
    \label{fig:num_states}
\end{figure}

\newpage

Although we have not been able to find closed-form formulae for $n \geq 2$ qubits so far,
there is a fairly complete picture for the $n=1$ case and we can explicitly derive the number of distinct Clifford+$kT$ states as follows.
First, recall the argument in Sec.~\ref{subsec:Clifford+kt_gates/states}:
\begin{lemma}\label{lem:strict_gate_vs_state}
	A one-qubit strict Clifford+$kT$ gate $U$ applied to $\ket{0}$ is a strict Clifford+$k'T$ state $(k' \leq k)$
	\begin{equation*}
		U\ket{0} = U'\ket{0}
	\end{equation*}
	(where $U'$ is a strict Clifford+$k'T$ gate) if and only if $U = U' T^m$ $(m\in \mathbb{Z}/8\mathbb{Z})$.
\end{lemma}

\noindent
Then, one is led to the following formula:

\begin{proposition} \label{prop:num of one-qubit CkT}
	The number of distinct one-qubit strict Clifford+$kT$ states $(k \geq 1)$ is
	\begin{equation*}
		n_k = 6 \cdot 2^k.
	\end{equation*}
\end{proposition}

\begin{proof}
	Employing the Matsumoto-Amano normal form
	of one-qubit strict Clifford+$kT$ gates~$U$, it follows from Lemma~\ref{lem:strict_gate_vs_state} that the number of distinct states of form $U\ket{0}$ is
	\begin{equation*}
		3 \cdot 2^{k-1} \cdot 24 \cdot \tfrac{1}{4}\ (= 9 \cdot 2^k)
	\end{equation*}
	since four $U$'s differing from each other by $(T^2 =)\, S$ gates in the rightmost Clifford gate $C$'s in Eq.~\eqref{eq:Matsumoto-Amano_normal_form}
	(and only those) give an identical state.\footnote{
		\setlength{\baselineskip}{16pt}%
        That there are $24$ one-qubit Clifford gates (up to global phase factors) in total can also be proved by the same reasoning:
        each of the $6$ one-qubit stabilizer states represents an equivalence class of one-qubit Clifford gates up to $S$ gates (which consists of $4$ gates).
	}
	The whole set of such states by definition contains all the strict Clifford+$kT$ states.
	
	On the other hand, for any strict Clifford$+(k-1)T$ state $U'\ket{0}$ where $U'$ is a strict Clifford$+(k-1)T$ gate,
	the gate $U'T$ must be either a strict Clifford+$kT$ gate or a strict Clifford$+(k-2)T$ gate from Proposition~\ref{prop:adding_T_one_qubit}, 
	but the latter possibility is inconsistent with the fact that $U'T\ket{0} = U'\ket{0}$ is a strict Clifford$+(k-1)T$ state and hence cannot occur.
	This conversely means that the set of $U\ket{0}$ states also contains all the strict Clifford$+(k-1)T$ states.	
	As Propositions~\ref{prop:adding_Clifford}, \ref{prop:adding_T_one_qubit} and Lemma~\ref{lem:strict_gate_vs_state} ensure that the set in fact does not contain any other states, one has
	\begin{equation*}
		n_k + n_{k-1} = 9 \cdot 2^{k}.
	\end{equation*}
	Given the fact $n_0 = 6$, the solution of this recurrence relation is straightforward.\footnotemark
\end{proof}

\newpage

\addtocounter{footnote}{-1}
\footnotetext{
	Division by $2^{k}$ gives
	\begin{equation*}
		\bigg(\frac{n_k}{2^k} - 6\bigg)
		=
		-\frac{1}{2} \cdot \bigg( \frac{n_{k-1}}{2^{k-1}} -6 \bigg)
	\end{equation*}
	and therefore
	\begin{equation*}
		\frac{n_k}{2^k} - 6 = \left(-\frac{1}{2}\right)^{k} \cdot (n_0 - 6) = 0.
	\end{equation*}
}

\noindent
In fact, it is indeed correctly reproduced in Table~\ref{tab:number_of_1-qubit_strict_Clifford+kT_states}. 
From this result, one can derive several normal forms for one-qubit strict Clifford+$kT$ states.

\begin{corollary}[State version of the normal form \textit{à la} \cite{MatsumotoAmano:2008}]
	A one-qubit strict Clifford+$kT$ state $(k \geq 1)$ can be uniquely expressed as
	\begin{equation*}
		\begin{array}{r}
			T\,(HT|SHT)\{k-1\}\ket{\phi},\phantom{\text{ or}}\\
			HT\,(HT|SHT)\{k-1\}\ket{\phi},\text{ or}\\
			SHT\,(HT|SHT)\{k-1\}\ket{\phi}\phantom{,\text{ or}}
		\end{array}
		\label{eq:1-qubit_strict_Clifford+kT_state_normal_form}
	\end{equation*}
	where $\ket{\phi} \in \mathcal{C}_1^{+0T}\backslash\{\ket{0}, \ket{1}\}$.
\end{corollary}

\begin{proof}
	These $3 \cdot 2^{k-1} \cdot 4$ states should include all the one-qubit strict Clifford+$kT$ states (note that $T\ket{0}=\ket{0}$ and $T\ket{1}=e^{i\frac{\pi}{4}}\ket{1} \sim \ket{1}$).
	Since there are exactly $3 \cdot 2^{k-1} \cdot 4 = 6\cdot 2^k$ distinct one-qubit strict Clifford+$kT$ states from Proposition~\ref{prop:num of one-qubit CkT}, the inclusion is actually an equality
	and the states are guaranteed to be distinct from each other.
\end{proof}

\begin{corollary}[State version of the normal form \textit{à la} \cite{GossetKliuchnikovMoscaRusso:2013}]
	A one-qubit strict Clifford+$kT$ state $(k \geq 1)$ can be uniquely expressed as
	\begin{equation*}
		R_{P_k}(+\tfrac{\pi}{4})R_{P_{k-1}}(+\tfrac{\pi}{4}) \cdots R_{P_1}(+\tfrac{\pi}{4}) \ket{\psi}
	\end{equation*}
	where $\ket{\psi} \in \mathcal{C}_1^{+0T}$, $P_i \neq I$, $P_{i+1} \neq P_{i}$, and $P_{1}\ket{\psi} \neq \pm \ket{\psi}$.
\end{corollary}

\begin{proof}
	These $2^{k} \cdot 6$ states should include all the one-qubit strict Clifford+$kT$ states,
	and the same argument as above applies.
\end{proof}

\newpage

\subsection{Clifford+$kT$ robustness of various non-stabilizer states} \label{subsec:numerical results of robustness}
Using the enumerated Clifford+$kT$ states, one can numerically compute the Clifford+$kT$ robustness of various quantum states 
by solving the basis pursuit problem Eq.~\eqref{eq:basis pursuit}. 
Since the problem size (the number of distinct Clifford+$kT$ states) grows rapidly with $n$ and $k$, we employ the technique introduced in \cite{HeinrichGross:2018} to speed up the computation of RoM by considering symmetries of a given quantum state, which is also applicable to computation of Clifford+$kT$ robustness with a slight modification.
For a more detailed explanation, see Appendix~\ref{app:symmetry}.

\begin{table}[h]
	\begin{tabular}{wc{40pt}|wc{6pt}wl{120pt}wl{140pt}wl{80pt}wl{40pt}}
		$k$		 && 0 & 1 & 2 & 3 \\ \hline
		$n=1$ && $1.4142136 \,\approx\, \sqrt{2}$						  & $1$																  		& $1$				 & $1$  \\
		$n=2$ && $1.7475469 \,\approx\, \frac{1 + 3\sqrt{2}}{3}$ & $1.3431458 \,\approx\, 7 - 4\sqrt{2}$		& $1$				  & $1$  \\
		$n=3$ && $2.2189514 \,\approx\, \frac{1 + 4\sqrt{2}}{3}$ & $1.7451660 \,\approx\, 47 - 32\sqrt{2}$		& $1.3431458$ & $1$  \\
		$n=4$ && $2.8627417 \,\approx\, \frac{3 + 8\sqrt{2}}{5}$ & $2.2161620 \,\approx\, 319 -224\sqrt{2}$	& $1.7451660$ & $\cdots$   \\
		$\vdots$
	\end{tabular}
	\caption{The values of $R_k((T\ket{+})^{\otimes n})$ with candidate symbolic expressions.}
	\label{tab:robustness_H}
\end{table}

First, we consider the Clifford+$kT$ robustness of $n$-qubit tensor products of the magic state $(T\ket{+})^{\otimes n}$.
The values of $R_k((T\ket{+})^{\otimes n})$ for small $n$ and $k$ are shown in Table~\ref{tab:robustness_H},
where some are identical for different $(n,k)$.
For example, the values for $(n,k) = (2,1)$ and $(3,2)$ are the same, as well as those for $(n,k) = (3,1)$ and $(4,2)$.
This is in fact not so surprising as they merely correspond to division of the system saturating the sub-multiplicativity (Theorem~\ref{prop:sub-multiplicative})\footnote{
	\setlength{\baselineskip}{16pt}%
	More generally, one can extend the inequality (in particular) for $k \leq n$ and see
	\begin{equation*}
		R_k((T\ket{+})^{\otimes n})
		\leq
		R_{k-1}((T\ket{+})^{\otimes (n-1)})
		\leq
		\cdots
		\leq
		R_{0}((T\ket{+})^{\otimes (n-k)}).
	\end{equation*}
}
\begin{equation} \label{eq:robustness_H_inequality}
	R_{k}((T\ket{+})^{\otimes n})
	\leq
	R_{k-1}((T\ket{+})^{\otimes (n-1)})
	\cdot
	\underbrace{
		R_{1}((T\ket{+})^{\otimes 1})
	}_{=1}.
\end{equation}
In other words, the optimal decomposition of $(T\ket{+})^{\otimes n}$ into Clifford+$kT$ states in those specific cases should be achieved by decomposing $(T\ket{+})^{\otimes (n-1)}$ into Clifford+$(k-1)T$ states, together with the (trivial) use of one $T$ gate to prepare a single $T\ket{+}$ state.

\newpage

More generally, the (almost-)equal values of $R_k((T\ket{+})^{\otimes n})$ for fixed $n-k$ combined with the argument in Sec.~\ref{subsec:sampling overhead}
tell us whether it makes any difference if we use Clifford+$kT$ states with $k\geq 1$ in addition to stabilizer states when simulating the tensor product of the magic states.
Assuming the knowledge of the optimal decompositions of $(T\ket{+})^{\otimes n'}$ into Clifford+$k'T$ pseudo-mixture for some $n' <n$ and $k' < k$
(and also that into stabilizer pseudo-mixture which should be easier to identify),
we want to compare the following two simulations of $(T\ket{+})^{\otimes n}$ using $k$ $T$ gates:
\begin{itemize}
\item each $T$ gate is used to prepare a $T\ket{+}$ state (resulting in $k$ of them in total)
\item every $k'$ $T$ gates is used to prepare a $(T\ket{+})^{\otimes n'}$ state (resulting in $\frac{k}{k'}$ of them in total)
\end{itemize}
For the latter to be cheaper than the former, Clifford+$kT$ robustness should satisfy
\begin{equation*}
	\big[
		R_1((T\ket{+})^{\otimes 1})
	\big]^k
	\big[
		R_0((T\ket{+})^{\otimes n'})
	\big]^{\frac{n-k}{n'}} 
	\overset{?}{>}
	\big[
		R_{k'}((T\ket{+})^{\otimes n'})
	\big]^{\frac{k}{k'}}
	\big[
		R_0((T\ket{+})^{\otimes n'})
	\big]^{\frac{n-\frac{k}{k'}n'}{n'}} 
\end{equation*}
or (since $R_1((T\ket{+})^{\otimes 1}) = 1$ and we are only interested in $n' > k'$ cases) equivalently
\begin{equation*}
	\big[
		R_0((T\ket{+})^{\otimes n'})
	\big]^{\frac{1}{n'}} 
	\overset{?}{>}
	\big[
		R_{k'}((T\ket{+})^{\otimes n'})
	\big]^{\frac{1}{n'-k'}}
\end{equation*}
where the right-hand side is almost the same as
$\big[
		R_{0}((T\ket{+})^{\otimes (n'-k')})
	\big]^{\frac{1}{n'-k'}}
$.
However, it appears that $\big[
	R_0((T\ket{+})^{\otimes n'})
\big]^{\frac{1}{n'}}$ is decreasing monotonically with $n'$ (see also \cite[Sec.\,4.1]{HeinrichGross:2018}) 
and therefore the desired inequality is typically not expected to hold,
meaning that using only stabilizer states would be enough.

\bigskip

\begin{table*}[h]
	\begin{tabular}{wc{40pt}|wc{6pt}wl{120pt}wl{120pt}wl{120pt}wl{80pt}}
		$k$		 && 0 & 1 & 2 & 3 \\ \hline
		$n=1$ && $1.7320508 \,\approx\, \sqrt{3}$								& $1.2247449 \,\approx\, \frac{\sqrt{6}}{2}$ & $1.0146119 \,\approx\, \frac{\sqrt{6}}{1+\sqrt{2}}$ & $1.0146119$\\
		$n=2$ && $2.2320508 \,\approx\, \frac{1 + 2\sqrt{3}}{2}$		   & $1.7941234$												  & $1.4245404$ & $1.1458531$\\
		$n=3$ && $3.0980762 \,\approx\, \frac{1 + 3\sqrt{3}}{2}$			& $2.5403969$												  & $2.0741375$ & $\cdots$  \\
		$n=4$ && $4.3310015 \,\approx\, \frac{13 + 20\sqrt{3}}{11}$ & $\cdots$   \\
		$\vdots$
	\end{tabular}
	\caption{The values of $R_k(\ket{SH}^{\otimes n})$ with candidate symbolic expressions.}
    \label{tab:robustness_SH}
\end{table*}

\newpage

Next, we consider the Clifford+$kT$ robustness of $n$-qubit tensor products of the magic state $\ket{SH}^{\otimes n}$.
We again observe in Table~\ref{tab:robustness_SH} the identical values for $(n,k) = (1,2)$ and $(1,3)$. However, this time the mechanism is completely different:
further numerical calculations in fact reveal that 
\begin{equation*}
	R_2(\ket{SH}^{\otimes 1})
	=
	R_3(\ket{SH}^{\otimes 1})
	=
	\cdots
	=
	R_6(\ket{SH}^{\otimes 1})
	>
	R_7(\ket{SH}^{\otimes 1}) \approx 1.0020233,
\end{equation*}
which correspond to saturation of the monotonicity (Theorem~\ref{prop:monotone}).
Our naive intuition is that the monotonicity is always strict,
but this example clearly shows that is not the case,
which also warns us to be careful when allocating limited number of $T$-gates in more generic and/or complicated cases.

We also note that the difference in values of $R_{k}(\ket{SH}^{\otimes n})$ for the same $n-k$ is larger than that in the $R_k((T\ket{+})^{\otimes n})$ case.
This indicates that a decomposition into Clifford+$kT$ pseudo-mixture is clearly advantageous over that into stabilizer pseudo-mixture in terms of sampling cost.
More concretely, a similar argument as in the $(T\ket{+})^{\otimes n}$ case asks whether Clifford+$kT$ robustness satisfies
\begin{equation*}
	\big[
		R_1(\ket{SH}^{\otimes 1})
	\big]^k
	\big[
		R_0(\ket{SH}^{\otimes n'})
	\big]^{\frac{n-k}{n'}} 
	\overset{?}{>}
	\big[
		R_{k'}(\ket{SH}^{\otimes n'})
	\big]^{\frac{k}{k'}}
	\big[
		R_0(\ket{SH}^{\otimes n'})
	\big]^{\frac{n-\frac{k}{k'}n'}{n'}}
\end{equation*}
or equivalently
\begin{equation*}
	\big[
		R_1(\ket{SH}^{\otimes 1})
	\big]^{k'}
	\big[
		R_0(\ket{SH}^{\otimes n'})
	\big]^{1-\frac{k'}{n'}} 
	\overset{?}{>}
	R_{k'}(\ket{SH}^{\otimes n'}),
\end{equation*}
and this time there exist cases where this holds (see Table~\ref{tab:robustness_SH_vs})
and hence the advantage is confirmed.

\begin{table}[h]
	\begin{tabular}{wc{40pt}|wc{6pt}wl{60pt}wl{60pt}wl{60pt}wl{60pt}}
		$k$		 && 0 & 1 & 2 & 3 \\ \hline
		$n=1$ && $1.732$ & $1.225$ & $0.866$ & $0.612$\\
		$n=2$ && $2.232$ & $1.830$ & $1.500$ & $1.230$ \\
		$n=3$ && $3.098$ & $2.603$ & $2.187$ & $\cdots$\\
		$n=4$ && $4.331$ & $\cdots$	&& \\
		$\vdots$
	\end{tabular}
	\caption{The values of $\big[
			R_1(\ket{SH}^{\otimes 1})
		\big]^{k}
		\big[
			R_0(\ket{SH}^{\otimes n})
		\big]^{1-\frac{k}{n}}$.
		The five values shown in the table for $k \geq 1$ and $n\geq 1$ are larger than $R_k(\ket{SH}^{\otimes n})$.
	}
	\label{tab:robustness_SH_vs}
\end{table}

\newpage

\begin{table*}[h]
	\begin{tabular}{wc{40pt}|wc{6pt}wl{90pt}wl{90pt}wl{90pt}wl{90pt}wl{40pt}}
		$k$		 && 0 & 1 & 2 & 3\\ \hline
		&
		& $2.2000000$
		& $1.7451660$
		& $1.3431458$
		& $1$ &\\
		&&
		& $\approx\, 47 - 32\sqrt{2}$
		& $\approx\, 7 - 4\sqrt{2}$
	\end{tabular}
	\caption{The values of $R_k(\mathrm{CS}\ket{+}^{\otimes 2})$ with candidate symbolic expressions.}
    \label{tab:robustness_CS}
\end{table*}

\begin{table*}[h]
	\begin{tabular}{wc{40pt}|wc{6pt}wl{90pt}wl{90pt}wl{90pt}wl{90pt}wl{40pt}}
		$k$		 && 0 & 1 & 2 & 3 & 4\\ \hline
		&
		& $2.5555555$
		& $2.2161620$
		& $1.7451660$
		& $1.3431458$
		& $1$\\
		&&
		& $\approx\, 319 -224\sqrt{2}$
		& $\approx\, 47 - 32\sqrt{2}$
		& $\approx\, 7 - 4\sqrt{2}$
	\end{tabular}
	\caption{The values of $R_k(\mathrm{CCZ}\ket{+}^{\otimes 3})$ with candidate symbolic expressions.}
    \label{tab:robustness_CCZ}
\end{table*}

Finally, we consider the Clifford+$kT$ robustness of other non-Clifford states that can be used as resource states for universal quantum computation.
Specifically, we consider a two-qubit state corresponding to the controlled-$S$ gate $\mathrm{CS}\ket{+}^{\otimes 2}$ and a three-qubit state corresponding to the controlled-controlled-$Z$ gate $\mathrm{CCZ}\ket{+}^{\otimes 3}$. 

The results are shown in Tables~\ref{tab:robustness_CS} and~\ref{tab:robustness_CCZ}.
One notices that the values of Clifford+$kT$ robustness for $k \geq 1$ are equal to those appearing in Table~\ref{tab:robustness_H}.
This is related to the realization of the corresponding non-Clifford gates using $T$ gates.
For example, it is known that a CCZ gate can be realized by using four $T$ gates~\cite{Jones:2012}.
Then, Corollary~\ref{corollary:gate_synthesis_no_go} implies that
\begin{equation} \label{eq:robustness_CCZ_inequality}
	R_{k+\Delta k}(\mathrm{CCZ}\ket{+}^{\otimes 3})
	\leq
	R_{k}((T\ket{+})^{\otimes {(4-\Delta k)}})
\end{equation}
holds for any $0 \leq \Delta k \leq 4$, and the identical values appearing in the two tables manifest actual saturation of this inequality.\footnote{
	\setlength{\baselineskip}{16pt}%
	As a by-product, from Eqs.~\eqref{eq:robustness_H_inequality} and \eqref{eq:robustness_CCZ_inequality} (taking $k=3, \Delta k =0$ in the latter) one has
	\begin{equation*}
		R_3(\mathrm{CCZ}\ket{+}^{\otimes 3})
		\leq
		R_3((T\ket{+})^{\otimes 4})
		\leq
		R_2((T\ket{+})^{\otimes 3}).
	\end{equation*}
	Combining this with a numerical result $R_3(\mathrm{CCZ}\ket{+}^{\otimes 3}) = R_2((T\ket{+})^{\otimes 3})$,
	one can determine $R_3((T\ket{+})^{\otimes 4})$ without (far more) costly computation.
}
Similar argument applies to the CS gate which is known to be realized by three $T$ gates. 

\newpage

\section{Outlook} \label{sec:conclusion}
In this study, we introduced the notion of Clifford+$kT$ robustness and examined its properties and applications through both theoretical and numerical analysis.
Several open questions remain for future investigation:
\begin{itemize}
\item (Sec.~\ref{subsec:lower bound}) Are there any $n$-qubit states $\rho$ such that $R_k(\rho)$ grows exponentially in $n$ even when $\frac{k}{n} > 1$?
If such states exist, they would provide new insights into the computational power of quantum circuits with limited $T$-gate resources.\footnote{
	\setlength{\baselineskip}{16pt}%
	The lower bound given in Theorem \ref{theorem:lower_bound} is not at all helpful for showing
	the (non-)existence of $n$-qubit states $\rho$ with exponentially-growing $R_k(\rho)$ for $\frac{k}{n} \geq 1$.
	This can be seen as follows.
	First, substituting the expansion in Pauli matrices $\rho = \frac{1}{2^n} \sum_{a=1}^{4^n} \Tr(\rho P_a) P_a$ into $\Tr\,\rho^2 \leq 1$ gives
	\begin{equation*}
		\frac{1}{2^n}  \sum_{a=1}^{4^n} [\Tr(\rho P_a)]^2 \leq 1. 
	\end{equation*}
	Then, \textit{e.g.} the Cauchy-Schwarz inequality gives an upper bound
	\begin{equation*}
		\sum_{a=1}^{4^n} |\Tr(\rho P_a)| \leq 
		(2\sqrt{2})^n,
	\end{equation*}
	where the equality holds when $|\Tr(\rho P_a)| = \frac{1}{\sqrt{2}}$ for all $P_a$ (but note that $\Tr(\rho I) = 1$ and such $\rho$ does not exist).
	Therefore, the lower bound of Clifford+$kT$ robustness given in Theorem~\ref{theorem:lower_bound} obeys
	\begin{equation*}
		\frac{1}{2^n(\sqrt{2})^k} \sum_a |\Tr(P_a\rho)| < (\sqrt{2})^{n-k},
	\end{equation*}
	which implies that (further) lower-bounding it (the left-hand side) will never provide us with useful information when $n-k \leq 0$ (or equivalently $\frac{k}{n} \geq 1$).
}
\item (Sec.~\ref{subsec:upper bound}) Is there a non-trivial upper bound on Clifford+$kT$ robustness as a function of $k$?
Since Clifford+$kT$ robustness characterizes the sampling cost of simulating a quantum circuit with $k$ $T$ gates, establishing such an upper bound would be valuable for estimating resource requirements in the early-FTQC era.
\item (Sec.~\ref{subsec:count of Clifford+kT}) Is there a closed-form expression for the number of distinct $n$-qubit Clifford+$kT$ states for $n \geq 2$?
A better understanding of the structure of the Clifford+$kT$ state space could lead to further theoretical development.
\item (Sec.~\ref{subsec:numerical results of robustness}) Can more numerical data on Clifford+$kT$ robustness be obtained for larger $n$ and $k$?
This would help deepen empirical understanding, and may be enabled by utilizing techniques for solving the basis pursuit problem, as proposed \textit{e.g.} in~\cite{hamaguchi2024handbookquantifying}.
\end{itemize}

\newpage

\section*{Acknowledgments}
The authors would like to thank Kohdai Kuroiwa for collaboration in the early stage of this work.

\bibliographystyle{ytamsalpha}
\bibliography{ref} 

\newpage

\appendix
\renewcommand{\thesubsection}{\thesection.\arabic{subsection}}

\section{Symmetry reduction in numerical calculations}\label{app:symmetry}
We describe how to exploit symmetries of a target quantum state to speed up calculation of its Clifford+$kT$ robustness, 
based on the technique introduced in~\cite{HeinrichGross:2018}. 

Let $G$ be a symmetry (sub)group of the target state $\rho_{\text{target}}$ with a representation on the Hilbert space
as conjugation by Clifford operators $\rho \mapsto C_g\,\rho\, C_g^\dagger$ $(g \in G)$.
Defining the corresponding symmetrization map as
\begin{equation*}
	\Pi_G :\rho \ \mapsto\ \frac{1}{|G|} \sum_{g \in G} C_g \,\rho\, C_g^\dag,
\end{equation*}
the space of (mixed) Clifford+$kT$ states is invariant under $\Pi_G$ (since $C_g$'s are Clifford operators)
and one can restrict the search space in the basis pursuit problem Eq.~\eqref{eq:basis pursuit} to the space of symmetrized Clifford+$kT$ states without loss of generality.
Since $C_g$-conjugate states are mapped to the same state under $\Pi_G$,
this restriction can drastically reduce the number of variables in the optimization problem, making numerical computations practical for larger values of $n$ and $k$.

Accordingly, we only need to construct representative (strict) Clifford+$kT$ states, not all of them.
Since 
\begin{equation*}
	R_P(+\tfrac{\pi}{4}) C_g\ket{\psi}
	=
	C_g R_{C_g^\dagger P C_g}(\pm\tfrac{\pi}{4}) \ket{\psi}
\end{equation*}
from Eq.~\eqref{eq:CR_PC=R_P'},
taking $\ket{\psi}$ to be a representative state of an equivalence class $\{C_g\ket{\psi}\}_{g\in G}$,
any state in $\{R_P(+\tfrac{\pi}{4}) C_g\ket{\psi}\}_{P,g}$ is either $C_g$-conjugate or $C_g\,R_{C_g^\dagger P C_g}(-\tfrac{\pi}{2})$-conjugate to
a state in $\{R_{P}(+\tfrac{\pi}{4}) \ket{\psi}\}_{P}$.
Therefore, recalling Proposition~\ref{prop:n-qubit_strict_Clifford+kT_state_normal_form},
one can construct representative (strict) Clifford$+(k+1)T$ states just by applying Pauli $(+\frac{\pi}{4})$-rotations to representative (strict) Clifford+$kT$ states.

In our calculations in Sec.~\ref{subsec:numerical results of robustness}, we considered permutation symmetries of qubit indices
(realized by combination of CNOT gates)
for states such as $(T\ket{+})^{\otimes n}$, $\ket{SH}^{\otimes n}$, and $\mathrm{CCZ}\ket{+}^{\otimes 3}$.  
In addition, for $(T\ket{+})^{\otimes n}$ (resp. $\ket{SH}^{\otimes n}$), local symmetries under $H$ (resp. $SH$) gates on each qubit are also taken into account.
These symmetries significantly reduced the basis-pursuit problem size. For instance, in the case of $(T\ket{+})^{\otimes 3}$, the number of representative Clifford+$3T$ states is reduced from 4,098,600 to 95,074.

\end{document}